\documentclass[11pt]{article}
\usepackage{amssymb,amsthm,amsmath}
\usepackage{graphicx}
\usepackage{epsfig,color}
\usepackage{epsf}
\usepackage[small]{caption}
\usepackage{amsfonts}
\usepackage{fullpage}
\usepackage{enumerate}
\usepackage{enumitem}
\usepackage{url}
\usepackage{cases}
\usepackage[update,prepend]{epstopdf}

\newtheorem{theorem}{Theorem}
\newtheorem{lemma}{Lemma}

\newcommand{\RR}{\mathbb{R}} 

\newcommand{\eps}{\varepsilon}

\def\N{\mathcal N}

\newcommand{\diam}{{\rm diam}}
\newcommand{\len}{{\rm len}}

\newcommand{\etal}{{et~al.}}
\newcommand{\ie}{{i.e.}}
\newcommand{\eg}{{e.g.}}

\newcommand{\pdv}[2]{\frac{\partial #1}{\partial #2}}

\newcommand{\later}[1]{}
\newcommand{\old}[1]{}

\title{\textsc{On the Longest Spanning Tree with Neighborhoods}}

\author{
Ke Chen\thanks{%
Department of Computer Science,
University of Wisconsin--Milwaukee, USA\@. 
Email:~\texttt{kechen@uwm.edu}}
\and
Adrian Dumitrescu\thanks{%
Department of Computer Science,
University of Wisconsin--Milwaukee, USA\@.
Email:~\texttt{dumitres@uwm.edu}}}

\begin{document}

\maketitle

\begin{abstract}
We study a maximization problem for geometric network design.
Given a set of $n$ compact neighborhoods in $\RR^d$, select a point in each neighborhood,
so that the longest spanning tree on these points (as vertices) has maximum length. 
Here we give an approximation algorithm with ratio $0.511$, 
which represents the first, albeit small, improvement beyond $1/2$.
While we suspect that the problem is NP-hard already in the plane,
this issue remains open.


\end{abstract}

\section{Introduction} \label{sec:intro}

In the \emph{Euclidean Maximum Spanning Tree Problem}  ({\sc EMST}),
given a set of points in the Euclidean space $\RR^d$, $d \geq 2$,
one seeks a tree that connects these points (as vertices) 
and has maximum length. The problem is easily solvable in polynomial time
by Prim's algorithm or by Kruskal's algorithm; algorithms that
take advantage of the geometry are also available~\cite{MPSY90}.

In this paper we study a natural generalization of the above problem.
In the \emph{Longest Spanning Tree with Neighborhoods} ({\sc Max-St-N}),
each point is replaced by a point-set, called \emph{neighborhood} (or \emph{region}),
and the tree must connect $n$ representative points, one chosen from each neighborhood
(duplicate representatives are allowed), and the tree has maximum length. 
The tree edges are straight line segments connecting pairs of points in distinct
neighborhoods. For obvious reasons we refer to these edges as \emph{bichromatic}.
As one would expect, the difficulty lies in choosing the representative points;
once these points are selected, the problem is reduced to the graph setting and
is therefore easily solvable. An example of a spanning tree for $10$ neighborhoods
is shown in Figure~\ref{fig:example1}. 

\begin{figure}[hbtp]
  \centering
  \includegraphics[scale=0.45]{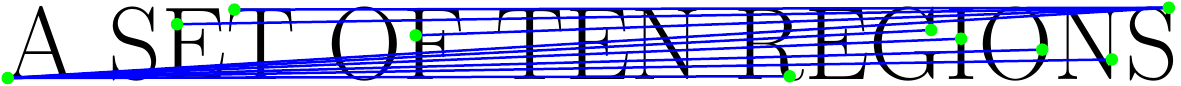}
\caption{An example of a long (still suboptimal) spanning tree for $10$ neighborhoods
  $\N=\{\text{A}, \text{S} \cup \text{S}, \text{E}\cup\text{E} \cup \text{E},
  \text{T} \cup \text{T}, \text{O} \cup \text{O}, \text{F}, \text{N} \cup \text{N},
  \text{R}, \text{G}, \text{I}\}$ (five neighborhoods are disconnected).
  The blue segments form a spanning tree on $\N$ and the green dots are the chosen
  representative points.}
\label{fig:example1}
\end{figure}

The input $\N$ consists of $n$ (possibly disconnected) neighborhoods. For simplicity,
it is assumed that each neighborhood is a union of polyhedra and the total vertex complexity
of the input is $N$. 

\paragraph{The greedy algorithm.}
A (natural) greedy algorithm chooses two points attaining a maximum inter-point distance
with points in distinct neighborhoods as representatives, and then repeatedly chooses
a point in another neighborhood as far as possible from some representative point.
The above algorithm does not necessarily find an optimal tree.
Let $\N=\{X_1,X_2,X_3\}$, where $X_1=\{a,c\}$, $X_2=\{b,c\}$, $X_3=\{d\}$, 
$\Delta{abc}$ is a unit equilateral triangle and $d$ is the midpoint of $ab$; 
see Figure~\ref{fig:example2}. 
Here the selection $a \in X_1$, $b \in X_2$, $d \in X_3$ yields a spanning tree in the form
of a star centered at $a$ of length $|ab| + |ad|=3/2$
(the edge lengths are $1$, $1/2$, and $1/2$ in the underlying complete graph).
On the other hand, selecting vertices $c \in X_1$, $c \in X_2$, $d \in X_3$ yields
a spanning tree in the form of a $2$-edge star centered at $d$
of length $|dc|+|dc|=2 \times \sqrt{3}/2 = \sqrt{3}$ which is better
(the edge lengths are $\sqrt{3}/2$, $\sqrt{3}/2$, and $0$ in the underlying complete graph).

\begin{figure}[hbtp]
  \centering
  \includegraphics[scale=0.45]{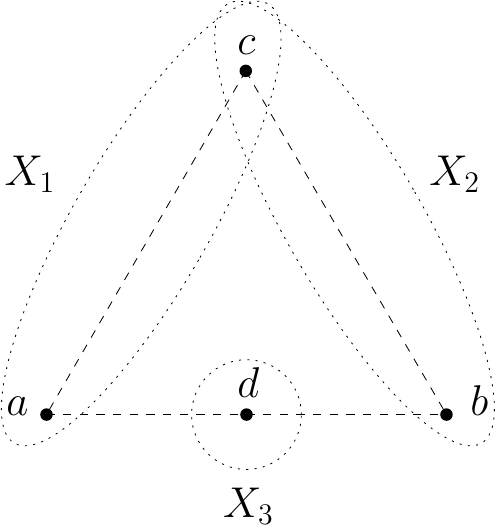}
\caption{Left: an example on which the greedy algorithm is suboptimal.}
\label{fig:example2}
\end{figure}

\paragraph{Definitions and notations.}
A {\em geometric graph} $G$ is a graph whose vertex set is a finite set of points in 
$\RR^d$ and whose edges consist of straight line segments~\cite[p.~223]{PA95}.
For two points $p, q\in \RR^d$, the Euclidean distance between them is denoted by $|pq|$.
The {\em length} of $G$, denoted by $\len(G)$, is the sum of the Euclidean lengths of all
edges in $G$. 

For a neighborhood $X \in \N$, let $V(X)$ denote its set of vertices.
Let $V=\cup_{X \in \N} V(X)$ denote the union of vertices of all neighborhoods in $\N$
and $N=|V|$. 

Given a set $\N$ of $n$ neighborhoods, we define the following parameters.
A \emph{monochromatic diameter} pair is a pair of points in the same neighborhood
attaining a maximum distance. 
A \emph{bichromatic diameter} pair is a pair of points from two neighborhoods
attaining a maximum distance, \ie, $p_i \in X_i$, $p_j \in X_j$, where
$X_i,X_j \in \N$, $i \neq j$, and $|p_ip_j|$ is maximum.
A \emph{diameter} pair is a pair of points (in the same neighborhood or in
different neighborhoods) attaining an overall maximum distance.

For $X \in \N$ and $p \in X$, let $d_\textrm{max}(p)$ denote the maximum distance
between $p$ and any point of a neighborhood $Y \in \N \setminus \{X\}$.
It is well known and easy to prove that both a monochromatic diameter and a bichromatic
diameter pair are attained by pairs of vertices in the input instance. 
An optimal (longest) spanning tree with neighborhoods is denoted by $T_\textrm{OPT}$;
it is a geometric graph whose vertices are the representative points of the $n$ neighborhoods.

\paragraph{Our results.}
We start by providing a factor $1/2$ approximation to {\sc Max-St-N}. We then
offer two refinement steps achieving a better ratio. The last refinement step
proves the following. 

\begin{theorem} \label{thm:1}
  Given a set $\N$ of $n$ neighborhoods in $\RR^d$ (with total vertex complexity $N$),
  a ratio $0.511$ approximation for the maximum spanning tree for the neighborhoods in $\N$
  can be computed in polynomial time. 
\end{theorem}

It is natural to try to include long edges (with endpoints in different neighborhoods)
when constructing a long spanning tree. However, in this regard
we show that every algorithm that always includes a bichromatic
diameter pair in the solution 
is bound to have an approximation ratio at most $\sqrt{2-\sqrt3}=0.517\ldots$
(via Figure~\ref{fig:eg} in Section~\ref{sec:analysis}).

\paragraph{Background and related work.}
Computing the minimum or maximum Euclidean spanning trees of a point set 
are classical problems in a geometric setting~\cite{MPSY90,PS85}.
The Traveling Salesman Problem ({\sc TSP}) is yet another related problem 
with a rich history of research in combinatorial optimization.
Several variants of the TSP including
the \emph{Euclidean Traveling Salesman Problem} ({\sc ETSP}) and
\emph{Maximum Traveling Salesman Problem} ({\sc MAX TSP})
are surveyed in~\cite{Ep00,Mi00,Mi17}.

While past research has primarily focused on minimization problems,
the maximization variants usually require different techniques and so they
are interesting in their own right and pose many unmet challenges.
See for instance the section devoted to longest subgraph problems in the survey of
Bern and Eppstein~\cite{BE97}.
The results obtained in this area in the last $20$ years are rather sparse;
the few articles~\cite{BFJ+03,DT10,Fe99} make a representative sample.
Recently, Biniaz~\etal~\cite{BBC+18} gave several approximation algorithms
for computing a longest noncrossing spanning tree in a multipartite geometric graph. 

Spanning trees for systems of neighborhoods have also been studied.
For instance, given a set of $n$ (possibly disconnected) compact neighborhoods in $\RR^d$,
select a point in each neighborhood so that the minimum spanning tree on these points
has minimum length~\cite{DFH+15,YLXX07}, or maximum length~\cite{DFH+15},
respectively. In the cycle version first studied by Arkin and Hassin~\cite{AH94},
called \emph{TSP with neighborhoods} ({\sc TSPN}), given a set of neighborhoods in $\RR^d$,
one needs to find a shortest closed curve (tour) intersecting each neighborhood.

\paragraph{Organization.} The rest of the paper is organized as follows.
Section~\ref{sec:approx} presents two approximation algorithms: one with ratio $0.5$
(in Subsection~\ref{ssec:a1}) and one with ratio $0.511$ (in Subsection~\ref{ssec:a2}).
The analysis of the latter algorithm is carried out in Section~\ref{sec:analysis}. 
We conclude in Section~\ref{sec:conclusion} with a summary of the results
in the context of related problems and some future research directions.

\section{Approximation Algorithms} \label {sec:approx}

For simplicity, we present our algorithms for the plane \ie, $d=2$.
The extension to higher dimensions is straightforward, and is briefly
discussed at the end.

Let $S=\{p_1,\ldots,p_n\}$, where $p_i=(x_i,y_i)$.
Given a point $p\in S$, the \emph{star centered at} $p$,
denoted by $S_p$, is the spanning tree on $S$ whose edges connect $p$ to the other points.
Using a technique developed in~\cite{DT10} (in fact a simplification of an earlier approach
from~\cite{ARS95}), we first obtain an approximation algorithm with ratio~$1/2$
({\tt Algorithm A1}). {\tt Algorithm A2} described later in this section implements
a refinement of this technique.

\subsection{A Simple $0.5$-Approximation Algorithm} \label{ssec:a1} 

\paragraph{{\tt Algorithm A1}.} 
Compute a bichromatic diameter of the point set $V$,
pick an arbitrary point (vertex) from each of the other $n-2$ neighborhoods, 
and output the longest of the two stars centered at one of the endpoints of the diameter.

\paragraph{Analysis.} Let $ab$ be a bichromatic diameter pair, and assume without loss
of generality that $ab$ is a horizontal unit segment, where $a=(0,0)$ and $b=(1,0)$.
We may assume that $a \in X_1$ and $b \in X_2$; refer to Figure~\ref{fig:f1}. 
\begin{figure}[hbtp]
  \centering
  \includegraphics[scale=0.5]{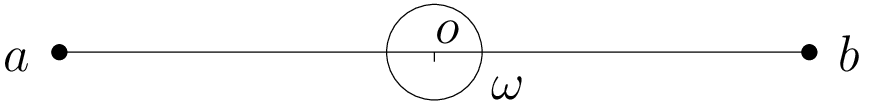}
\caption{A bichromatic diameter pair $a,b$ and the disk $\omega$.}
\label{fig:f1}
\end{figure}

The ratio $1/2$ (or $\frac{n}{2n-2}$ which is slightly better) follows from the next lemma
in conjunction with the obvious upper bound 
\begin{equation} \label{eq:trivial-ub}
\len(T_\textrm{OPT}) \leq n-1.
\end{equation}
The latter is implied by the fact that each edge of $T_\textrm{OPT}$ is bichromatic
and thus of length at most~$1$.

\begin{lemma} \label{lem:2stars}
Let $S_a$ and  $S_b$ be the stars centered at the points $a$ and $b$,
respectively. Then $ \len(S_a) + \len(S_b) \geq n$.
\end{lemma}
\begin{proof}
Assume that $a=p_1$, $b=p_2$. For each $i=3,\ldots,n$, the triangle inequality
for the triple $a,b,p_i$ gives
$$ |ap_i| + |b p_i| \geq |ab|=1. $$
By summing up we have
\begin{equation*}
\len(S_a) + \len(S_b)= \sum_{i=3}^n (|ap_i| + |b p_i|) + 2|ab|
\geq (n-2)+2=n.
\qedhere
\end{equation*}
\end{proof}

\subsection{An Improved $0.511$-Approximation Algorithm} \label{ssec:a2} 

We next refine the previous algorithm to achieve an approximation ratio of $0.511$.
The setting is the same, where $ab$ is a bichromatic diameter pair of unit length.
The technique uses two parameters $x$ and $y$, introduced below. The smallest value
of the ratio obtained over the entire range of admissible $x$ and $y$ is determined
and yields the approximation ratio of {\tt Algorithm~A2}.

Let $o$ be the midpoint of $ab$, and $\omega$ be the disk centered at $o$,
of minimum radius, say, $x$, containing at least $\lfloor n/2 \rfloor$
of the neighborhoods $X_3,\ldots,X_n$. In particular, this implies that
we can consider $\lfloor n/2 \rfloor$ neighborhoods as contained in $\omega$
and $\lceil n/2 \rceil$ neighborhoods having points on the boundary $\partial \omega$
or in the exterior of  $\omega$. We first argue that for $x \geq 0.2$,
the $0.511$ approximation ratio easily follows (with room to spare).
Observe that for each of the neighborhoods not contained in $\omega$, one of the connections
from an arbitrary point of the neighborhood to $a$ or $b$ is at least $\sqrt{\frac14 +x^2}$.
Let $T$ be the spanning tree consisting of all such longer connections together with $ab$.
Then,

\begin{equation} \label{eq:T}
  \len(T) \geq 1+ \left \lfloor \frac{n}{2} \right \rfloor \frac12 +
  \left( \left\lceil \frac{n}{2} \right\rceil -2\right) \sqrt{\frac14 +x^2}.
\end{equation}

The above expression can be simplified as follows.
{\large
\begin{equation} \label{eq:T-even-odd}
\len(T) \geq 
\begin{cases}
1+ \frac{n}{4} + \left(\frac{n}{4} -1 \right) \sqrt{1+4x^2}, &
\text{ if } n \text{ is even},  \vspace{8pt} \\

1+ \frac{n-1}{4} + \left( \frac{n+1}{4} -1 \right) \sqrt{1+4x^2}, &
\text{ if } n \text{ is odd}. 
\end{cases}
\end{equation}
}

Assume first that $x \geq 1/2$. Then,
\begin{align} \label{eq:T:even1}
\len(T) &\geq 1+ \frac{n}{4} + \left(\frac{n}{4} -1 \right) \sqrt{2}
=\left(\frac14 + \frac{\sqrt2}{4} \right) (n-1) + \left(\frac{5-3\sqrt2}{4}\right)
\nonumber \\
&\geq \left(\frac{1+ \sqrt2}{4} \right) (n-1),
  \text{ if } n \text{ is even}, \text{ and}\\ 
\len(T) &\geq   1+ \frac{n-1}{4} + \left( \frac{n+1}{4} -1 \right) \sqrt{2}
=\left(\frac14 + \frac{\sqrt2}{4} \right) (n-1) + \left( 1 - \frac{\sqrt2}{2} \right)
\nonumber \\
&\geq \left(\frac{1+\sqrt2}{4} \right) (n-1),
\text{ if } n \text{ is odd}. 
\end{align}   

Together with~\eqref{eq:trivial-ub}, we have $\len(T) \geq 0.6 \, \len(T_\textrm{OPT})$,
for every $n \geq 2$ and $x \geq 1/2$.

\medskip
Assume next that $x \leq 1/2$. 
If $n$ is even, \eqref{eq:T-even-odd} yields
\begin{align*}
\len(T) &\geq 1+ \frac{n}{4} + \left(\frac{n}{4} -1 \right) \sqrt{1+4x^2}\\
&=\frac{n-1}{4}\left(1+\sqrt{1+4x^2}\right) + \left(\frac54 - \frac34 \sqrt{1+4x^2}\right)\\
&\geq \frac{n-1}{4}\left(1+\sqrt{1+4x^2}\right).
\end{align*}
Indeed, the last term in the second line is non-negative for $x \leq 1/2$.

If $n$ is odd, \eqref{eq:T-even-odd} yields
\begin{align*}
\len(T) &\geq 1+ \frac{n-1}{4} + \left( \frac{n+1}{4} -1 \right) \sqrt{1+4x^2}\\ 
&=\frac{n-1}{4} \left(1+\sqrt{1+4x^2}\right) + \left(1 - \frac12 \sqrt{1+4x^2}\right)\\
&\geq \frac{n-1}{4} \left(1+\sqrt{1+4x^2}\right).
\end{align*}
Again, the last term  in the second line is non-negative for $x \leq 1/2$.

Consequently, for every $n \geq 2$ and $x \leq 1/2$ we have
\begin{equation} \label{eq:T:x}
\len(T) \geq \frac{n-1}{4}\left(1+\sqrt{1+4x^2}\right).
\end{equation}

It is easy to check that
\begin{equation} \label{eq:0.519}
  \frac{1+\sqrt{1+4x^2}}{4}  \geq \frac{5+\sqrt{29}}{20}= 0.519\ldots,
  \text{ for } x \geq 0.2. 
\end{equation}

Hence the approximation ratio is at least $0.519\ldots$ if $x \geq 0.2$.
We therefore subsequently assume that $x \leq 0.2$.
Let the monochromatic diameter of $V$ be $1+y$, for some $y \in [-1,\infty)$.
The next lemma shows that $y \leq 1$, and so the monochromatic diameter of $V$
is $1+y$, for some $y \in [-1,1]$.

\begin{lemma} \label{lem:y}
For every $X \in \N$, $\diam(X) \leq 2$.
\end{lemma}
\begin{proof}
  Let $pq$ be a diameter pair of neighborhood $X$. 
  Let $r$ be an arbitrary point of an arbitrary neighborhood $Y \in \N \setminus \{X\}$.
  By the triangle inequality, we have $|pq| \leq |pr|+|rq| \leq 1+1=2$, as required.
\end{proof}

  If $y \geq 0.2$, let $a_1,b_1 \in X$ be a corresponding diameter pair.
  Choose a point in every other neighborhood and connect it to $a_1$ and $b_1$. 
  Since $|a_1 b_1|= 1+y \geq 1.2$, the longer of the two stars centered
  at $a_1$ and $b_1$ has length at least $(n-1)(1+y)/2 \geq 0.6(n-1)$; this
  candidate spanning tree offers thereby this ratio of approximation.
  We will subsequently assume that $y \in [-1,0.2]$.

  We have thus shown that a constant approximation ratio better than $0.511$ can be obtained
  if $x$ or $y$ is sufficiently large. In the complementary case, \ie, both $x$ and $y$ are small, 
  we apply the following algorithm.

\paragraph{{\tt Algorithm A2}.} 
The algorithm computes two candidate solutions $T_1$ and $T_2$ and returns the best
of the two. The setting is the same as in Subsection~\ref{ssec:a1}.
In particular, it is assumed that $x \in [0,0.2]$ and $y \in [-1,0.2]$
(outside this range, the approximation ratio exceeds $0.519$). 

The first candidate solution $T_1$ for the spanning tree is only relevant
for the range $y \geq 0$ (if $y<0$ its length could be smaller than $(n-1)/2$
and $T_1$ will be ignored). 
Suppose that a monochromatic diameter pair in $V$ is achieved by a pair
$a_1,b_1 \in X$. Recall that $|a_1b_1|=1+y$.
Choose an arbitrary point in every other neighborhood and connect it to $a_1$ and $b_1$. 
Let $T_1$ be the longer of the two stars centered at $a_1$ and $b_1$.

The second candidate solution $T_2$ for the spanning tree connects each of the
neighborhoods contained in $\omega$ with either $a$ or $b$ at a cost of at least $1/2$
(based on the fact that $\max\{|a p_i|, |b p_i|\} \geq |ab|/2=1/2$). 
For each neighborhood $X_i$, $i \geq 3$, select the vertex of $X_i$ that is farthest from $o$
and connect it with $a$ or $b$, whichever yields the longer connection.
As such, if $X_i$ is not contained in $\omega$, the connection length is at least
$\sqrt{\frac14 +x^2}$. Finally add the unit segment $ab$.

\paragraph{Lower bounds on the lengths of candidate solutions.} 
By the triangle inequality, the length of the star $T_1$ is bounded from below as follows:
\begin{equation} \label{eq:T1}
\len(T_1) \geq (n-1) \, \frac{1+y}{2}. 
\end{equation}

The length of $T_2$ is bounded from below by~\eqref{eq:T:x}. As such, we have
\begin{equation} \label{eq:T2}
\len(T_2) \geq \frac{n-1}{4}\left(1+\sqrt{1+4x^2}\right).
\end{equation}

\smallskip
  In order to prove the claimed approximation ratio $0.511$ for {\tt Algorithm A2},
  we first derive a sharper upper bound on the length of $T_\textrm{OPT}$
  when both $x$ and $y$ are smaller than $0.2$.

\subsection{Upper bound on $\len(T_\textrm{OPT})$} \label{ssec:ub-topt}

Let $\Omega$ be the disk of radius $R(y)$ centered at $o$, where
\[ R(y)= \begin{cases} \frac{\sqrt3}{2} &\mbox{if } y \leq 0 \\
\frac{\sqrt3}{2} +\frac{2}{\sqrt3} y & \mbox{if }  y \geq 0
\end{cases}
\]

\begin{lemma} \label{lem:Omega}
$V$ is contained in $\Omega$.
\end{lemma}
\begin{proof}
  Assume for contradiction that there exists a point $p_i \in X_i$ at distance larger than
  $R(y)$ from $o$. By symmetry, we may assume that 
  $|ap_i| \leq |bp_i|$ and that $p_i$ lies in the closed halfplane above the line containing $ab$.

  First consider the case $y \leq 0$; it follows that
  $|b p_i| > \sqrt{\frac14 + \frac34} =1$.
  If $i=2$, then $b,p_i \in X_2$, which contradicts the definition of $y$.
  Otherwise, $b \in X_2$ and $p_i \in X_i$ are points in different neighborhoods at distance
  larger than $1$, in contradiction with the original assumption on the bichromatic
  diameter of $V$.   

  Next consider the case $y \geq 0$; it follows that
  $|b p_i| \geq \sqrt{\frac14 + \left(\frac{\sqrt3}{2} +\frac{2}{\sqrt3} y\right)^2} > 1+y$.
  If $i=2$, then $b,p_i \in X_2$, which contradicts the definition of $y$.
  Otherwise, $b \in X_2$ and $p_i \in X_i$ are points in different neighborhoods at distance
  larger than $1$, in contradiction with the original assumption on the bichromatic
  diameter of $V$.

 In either case (for any $y$) we have reached a contradiction, and this concludes the proof.
\end{proof}

Recall that for a point $p\in X\in \N$,
$d_\textrm{max}(p)$ is the maximum distance between $p$ and a point of a neighborhood 
$Y \in \N \setminus \{X\}$.
\begin{lemma} \label{lem:ub:generic}
  Let $\N=\{X_1,\ldots,X_n\}$ be a set of $n$ neighborhoods and $T_\textrm{OPT}$ be an
  optimal spanning tree assumed to connect points (vertices) $p_i \in X_i$ for $i=1,\ldots,n$. 
  For every $j \in [n]$, we have 
$$ \len(T_\textrm{OPT}) \leq \sum_{i \neq j} d_\textrm{max}(p_i). $$
\end{lemma}
\begin{proof}
Consider $T_\textrm{OPT}$ rooted at $p_j$. Let $\pi(v)$ denote the
parent of a (non-root) vertex $v$. Uniquely assign each
edge $\pi(v) v$ of $T_\textrm{OPT}$ to vertex $v$. The inequality 
$\len(\pi(v) v) \leq d_\textrm{max}(v)$ holds for each edge of the tree.
By adding up the above inequalities,
the lemma follows.
\end{proof}
\begin{lemma} \label{lem:3}
  If $X \in \N$ is contained in $\omega$, and $p \in X$, then
  $d_\textrm{max}(p) \leq \min(1,x+R(y))$. 
\end{lemma}
\begin{proof}
  By definition,  $d_\textrm{max}(p) \leq 1$. 
  By Lemma~\ref{lem:Omega}, the vertex set $V$ is contained in $\Omega$
  and thus all neighborhoods in $\N$ are contained in $\Omega$.
  By the triangle inequality,
  $d_\textrm{max}(p) \leq |po| + R(y) \leq x+R(y)$,
  as claimed.
\end{proof}
\begin{lemma} \label{lem:ub}
The following inequality holds:
\begin{equation} \label{eq:ub}
  \len(T_\textrm{OPT}) \leq (n-1) \cdot \min \left(1, \frac{1 + x +R(y)}{2} \right). 
\end{equation}
\end{lemma}
\begin{proof}
  Let $T_\textrm{OPT}$ be a longest spanning tree of $p_1,\ldots,p_n$, where
  $p_i \in X_i$, for $i=1,\ldots,n$. View $T_\textrm{OPT}$ as rooted at $p_1 \in X_1$;
  recall that $a \in X_1$. By Lemma~\ref{lem:ub:generic},
  $$ \len(T_\textrm{OPT}) \leq \sum_{i=2}^n d_\textrm{max}(p_i). $$
  If $X_i$ is not contained in $\omega$, $d_\textrm{max}(p_i) \leq 1$;
  otherwise, by Lemma~\ref{lem:3}, $d_\textrm{max}(p_i) \leq \min(1,x+R(y))$. 
  By the setting of $x$ in the definition of $\omega$, we have
\begin{equation} \label{eq:ub1}
  \len(T_\textrm{OPT}) \leq \left(\left \lceil \frac{n}{2} \right \rceil -1\right) \cdot 1
  + \left \lfloor \frac{n}{2} \right \rfloor  \cdot \min(1,x+R(y)).
\end{equation}
If $n$ is even, Inequality~\eqref{eq:ub1} yields (since the second term in the second line
is $\leq 0$)
\begin{align*}
\len(T_\textrm{OPT}) &\leq \left(\frac{n}{2} -1\right) + \frac{n}{2} \cdot \min(1,x +R(y)) \\
&\leq \frac{n-1}{2} \left(1+x+R(y)\right) + \frac{\min(1,x+R(y))-1}{2}\\
&\leq \frac{n-1}{2} \left(1+x+R(y)\right). 
\end{align*}

If $n$ is odd, Inequality~\eqref{eq:ub1} yields
\begin{align*}
\len(T_\textrm{OPT}) &\leq \frac{n-1}{2} + \frac{n-1}{2} \left( x +R(y) \right)
= \frac{n-1}{2} \left(1+x+R(y)\right).
\end{align*}
Therefore the above inequality holds for every $n \geq 2$. The lemma follows
by adjoining the trivial upper bound in equation~\eqref{eq:trivial-ub}.
\end{proof}

\section{Analysis of {\tt Algorithm A2}} \label {sec:analysis}

We start with a preliminary argument for ratio $0.506$ that comes with a simpler proof.
We then give a sharper analysis for ratio $0.511$.

\paragraph{A preliminary estimate on the approximation ratio of {\tt Algorithm A2}.}
First consider the case $y<0$. Then $R(y)=\sqrt{3}/2$, so the ratio of {\tt Algorithm A2}
is at least
\begin{align*}
\min_{\small \begin{array}{c} 0 \leq x \leq 0.2\\y<0 \end{array}}
\frac{\len(T_2)}{\len(T_{\textrm{OPT}})}
&\geq \min_{\small \begin{array}{c} 0 \leq x \leq 0.2\end{array}}
\frac{1+\sqrt{1+4x^2}}{\min\left(4, 2+\sqrt{3}+2x\right)}.
\end{align*}
A standard analysis shows that this ratio achieves its minimum 
$\left(1+2\sqrt{2-\sqrt{3}}\middle)\right/4=0.508\ldots$ when $x=1-\sqrt{3}/2$.

When $y\geq 0$, the ratio of {\tt Algorithm A2} is at least
$$ \min_{\small \begin{array}{c} 0 \leq x,y \leq 0.2 \end{array}} 
\max \left( \frac{\len(T_1)}{\len(T_\textrm{OPT})}, \,
\frac{\len(T_2)}{\len(T_\textrm{OPT})} \right). $$

The inequalities~\eqref{eq:T1}, \eqref{eq:T2}, \eqref{eq:ub}
imply that this ratio is at least
\[
  \frac{\max \left( 1+y, (1 + \sqrt{1+4x^2})/2 \right)}
  {\min \left(2, 1+x +R(y) \right)} = \frac{\max \left( 1+y, (1 + \sqrt{1+4x^2})/2 \right)}
  {\min \left(2, 1 +\frac{\sqrt3}{2} +x + \frac{2}{\sqrt3}\,y \right)}.
\]
Since the analysis is similar to that for deriving the refined bound we give next,
we state without providing details that this piecewise function reaches its minimum
value
\[ \left(4\sqrt{3}-1-2\sqrt{9-3\sqrt{3}}\middle)\right/4=0.506\ldots \]
when
\[ y=\left(4\sqrt{3}-3-2\sqrt{9-3\sqrt{3}}\middle)\right/2=0.0137\ldots 
\text { and } x=\sqrt{3}/2-3+2\sqrt{3-\sqrt{3}}=0.1180\ldots \]
This provides a preliminary ratio 0.506 in Theorem~\ref{thm:1}.

\paragraph{A refined bound.}
Let $m=\lfloor n/2\rfloor$. 
Assume for convenience that the neighborhoods $X_3,\ldots,X_n$ are relabeled so that
$X_3,\ldots,X_{m+2}$ are contained in $\omega$ and 
$X_{m+3},\ldots,X_n$ are not contained in the interior of $\omega$.
Recall that $p_i \in X_i$ are the representative points in an optimal solution $T_\textrm{OPT}$.
Let $x_i =|op_i|$, for $i=3,\ldots,m+2$; as such, $x_3,\ldots,x_{m+2} \leq x$. 
Denote the average of $x_3,\ldots,x_{m+2}$ by $z$, \ie, $\sum_{i=3}^{m+2} x_i =m z$, 
and note that $z \leq x$.

As in the proof of Lemma~\ref{lem:3}, by the triangle inequality we have
$$ d_\textrm{max}(p_i) \leq |op_i| + R(y) = x_i+R(y), \text{ for } i=3,\ldots,m+2. $$
Consequently, the upper bound in~\eqref{eq:ub} can be improved to
\begin{equation} \label{eq:ub:new}
  \len(T_\textrm{OPT}) \leq (n-1) \cdot \min \left(1, \frac{1 + z +R(y)}{2} \right). 
\end{equation}

We next obtain an improved lower bound on $\len(T_2)$.
Recall that {\tt Algorithm A2} selects the vertex of $X_i$ that is farthest from $o$
for every $i \geq 3$, and connects it with $a$ or $b$, whichever yields the longer connection. 
In particular, the length of this connection is at least $\sqrt{\frac14 +x_i^2}$ for
$i=3,\ldots,m+2$. Let $h \colon [0,\infty) \to \RR$ be defined as follows
\begin{equation} \label{eq:h}
h(x) = \sqrt{1+4x^2}.
\end{equation}
It is easy to check that
\begin{equation} \label{eq:h'h''}
h'(x) = \frac{4x}{\sqrt{1+4x^2}} \geq 0  \text{  and  } 
h''(x) = \frac{4}{(1+4x^2)^{3/2}}>0.
\end{equation}
As such, the function $h(x)$ is convex, \ie,
\begin{equation} \label{eq:convex}
h\left(\frac{x+y}{2}\right) \leq \frac{h(x) + h(y)}{2}, \text{ for every } x,y \geq 0.
\end{equation}
and Jensen's inequality yields:
\begin{equation} \label{eq:Jensen}
  \sum_{i=3}^{m+2} \sqrt{1+4x_i^2} \geq m \sqrt{1+4z^2 }.
\end{equation}

We thereby obtain (noting that $z \leq x$) the following sharpening of the
lower bound in~\eqref{eq:T2}:
\begin{equation} \label{eq:lb:new}
  \len(T_2) \geq \frac{n-1}{4}\left( \sqrt{1+4z^2} + \sqrt{1+4x^2} \right)
  \geq \frac{n-1}{2}\sqrt{1+4z^2}.
\end{equation}

To analyze the approximation ratio we relate the upper bound~\eqref{eq:ub:new} to 
the lower bound~\eqref{eq:lb:new} and distinguish two cases:

\medskip
\emph{Case 1:} $y \leq 0$. 
Then $R(y)=\sqrt{3}/2$, 
so the ratio of {\tt Algorithm A2} is at least
\[
\min_{\small \begin{array}{c} 0 \leq z \leq 0.2 \end{array}} 
\frac{\len(T_2)}{\len(T_\textrm{OPT})} \geq
\min_{\small \begin{array}{c} 0 \leq z \leq 0.2 \end{array}} 
\frac{2\sqrt{1+4z^2}}
{\min \left(4, 2+ 2z + \sqrt{3} \right)}.
\]
When $4\leq 2+2z + \sqrt{3}$, we have $z \geq 1-\sqrt{3}/2$.
Then
\[
\frac{\sqrt{1+4z^2}}{2}\geq 
\frac{\sqrt{8-4\sqrt{3}}}{2}=\sqrt{2-\sqrt{3}}=0.517\ldots.
\]
When $2+2z + \sqrt{3} \leq 4$, \ie, $z \leq 1-\sqrt{3}/2$, let
$$ f(z)=\frac{2\sqrt{1+4z^2}}{2 + \sqrt3 + 2z}. $$
Then 
\[
f'(z) = \frac{8\left(2+\sqrt{3} \right)z-4}{\sqrt{1+4z^2}\left(2+\sqrt{3}+2z \right)^2}.
\]
Since 
$8\left(2+\sqrt{3} \right)z-4\leq 4\left(2+\sqrt{3} \right)\left(2-\sqrt{3} \right)-4=0$,
the function is non-increasing on $[0,1-\sqrt{3}/2]$ and so
\[
f(z)\geq f\left(1-\sqrt{3}/2\right)=\sqrt{2-\sqrt{3}}=0.517\ldots.
\]
This concludes the proof for the first case.

\medskip
\emph{Case 2:} $y \geq 0$, then the ratio of {\tt Algorithm A2} is at least
$$ \min_{\small \begin{array}{c} 0 \leq y,z \leq 0.2\end{array}} 
\max \left( \frac{\len(T_1)}{\len(T_\textrm{OPT})}, \,
\frac{\len(T_2)}{\len(T_\textrm{OPT})} \right). $$
For $0 \leq y, z \leq 0.2$, let
\[
g(z,y)=\frac{\max \left( 1+y, \sqrt{1+4z^2} \right)}{\min \left(2, 1+ z +R(y) \right)}
=\frac{\max \left( 1+y, \sqrt{1+4z^2}\right)}
{\min \left(2, 1 +\frac{\sqrt3}{2} +z + \frac{2}{\sqrt3}\,y \right)}.
\]

The inequalities~\eqref{eq:T1}, \eqref{eq:ub:new}, \eqref{eq:lb:new}
imply that the ratio of {\tt Algorithm A2} is at least
\[ \min_{\small \begin{array}{c} 0 \leq y,z \leq 0.2\end{array}}  g(z,y). \]

The curve $\gamma: 1+y=\sqrt{1+4z^2}$ and
the line $\ell: 2 = 1 +\frac{\sqrt3}{2} +z + \frac{2}{\sqrt3}\,y$
split the feasible region $[0, 0.2] \times [0, 0.2]$ into four subregions;
see Figure~\ref{fig:analysis2}.
\begin{figure}[hbtp]
  \centering
  \includegraphics[scale=0.4]{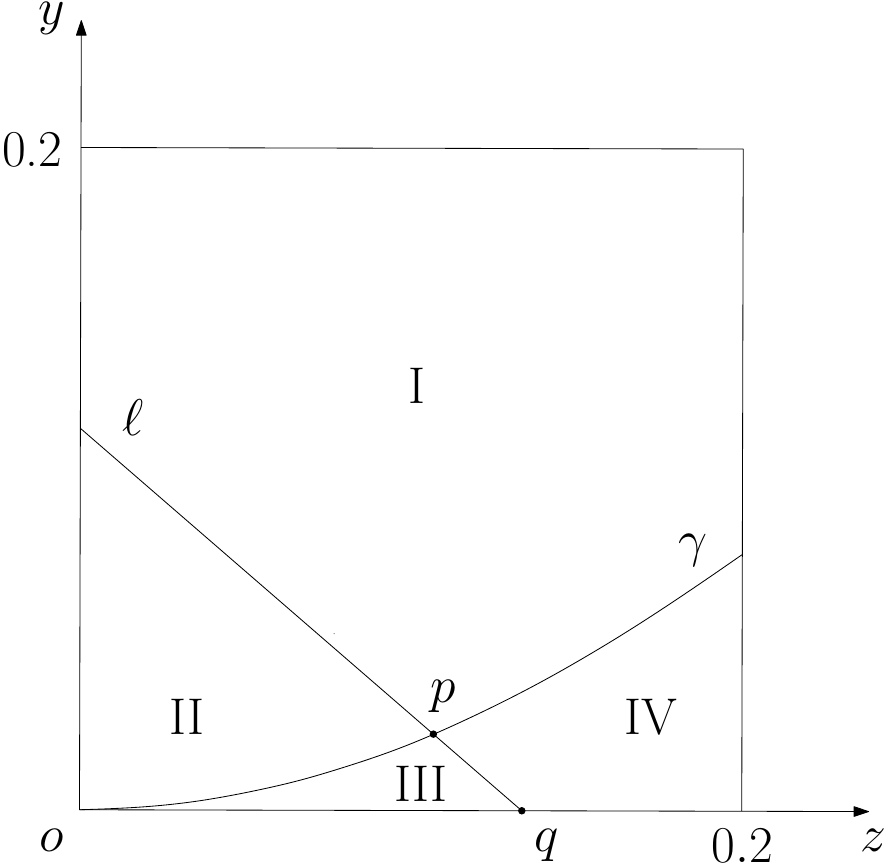}
\caption{The feasible region of the function $g(z,y)$.}
\label{fig:analysis2}
\end{figure}
The curve $\gamma$ intersects line $\ell$ at point $p=\left(z_0, y_0\right)$, where 
$z_0=\left(8\sqrt[4]{3}-\sqrt{3}-6\middle)\right/26=0.1075\ldots$ and
$y_0=\left(8\sqrt{3}-2\sqrt[4]{27}-9\middle)\right/13=0.0228\ldots$ 
Set
\begin{equation} \label{eq:rho}
\rho:= (1+y_0)/2 = \left(4\sqrt{3}+2-\sqrt[4]{27}\middle)\right/13= 0.511\ldots
\end{equation}

In region I, $g(z,y)=(1+y)/2$. It reaches the minimum value $\rho$
when $y$ is minimized, \ie, $y=y_0$.

In region II, $g(z,y)=\dfrac{1+y}{1 +\sqrt{3}/2 +z + 2y/\sqrt{3}}$. Its partial derivative
is positive, \ie,
$$ \pdv{g}{y} =\frac{1-\sqrt{3}/6+z}{\left(1+\sqrt{3}/2+z+2y/\sqrt{3}\right)^2}>0, $$
so $g(z,y)$ reaches its minimum value on the curve $\gamma$.
On this curve, let
\[ G(z)= g\left(z,y(z)\right)=\frac{\sqrt{1+4z^2}}{1-\sqrt{3}/6+z+2\sqrt{1+4z^2}/\sqrt{3}}. \]
Its derivative is
\[ G'(z)=\frac{\left(4-2\sqrt{3}/3\right)z-1}
   {\sqrt{1+4z^2}\left(1-\sqrt{3}/6+z+2\sqrt{1+4z^2}/\sqrt{3}\right)^2}. \]
   Note that the numerator of $G'(z)$ is negative, \ie,
   $\left(4-2\sqrt{3}/3\right)z-1<4z-1<0$ for $z\in[0, 0.2]$, thus $G'(z)<0$.
So the minimum value is $\rho$, and is achieved when $z$ is maximized, \ie, $z=z_0$.

In region IV, $g(z,y)=\sqrt{1+4z^2}/2$ which increases monotonically with respect to $z$.
So the minimum value is again $\rho$ and is achieved when $z$ is minimized, \ie, $z=z_0$.

In region III,
$$g(z,y)=\frac{\sqrt{1+4z^2}}{1 +\sqrt{3}/2 +z + 2y/\sqrt{3}}. $$
Its partial derivative is negative, \ie, 
$$ \pdv{g}{y}=\frac{-2\sqrt{1+4z^2}}{\sqrt{3}\left(1 +\sqrt{3}/2 +z + 2y/\sqrt{3}\right)^2}<0,$$ 
so $g(z,y)$ reaches its minimum value on the arc $op \subset \gamma$ or the segment
$pq \subset \ell$,
where $q=(1-\sqrt{3}/2, 0)$ is the intersection point of $\ell$ and the $z$-axis. 
Since these two curves are shared with region II and IV respectively, by previous analyses,
$g(z,y)$ reaches its minimum value $\rho$ at point $p$.

In summary, we showed that
$$ \min_{\small \begin{array}{c} 0 \leq y, z \leq 0.2 \end{array}} g(z,y) \geq \rho= 0.511\ldots, $$
establishing the approximation ratio in Theorem~\ref{thm:1}.
\qed

\paragraph{Remarks.} 1. The algorithm can be adapted to work in $\RR^d$ for any $d \geq 3$.
In the analysis, the disk $\omega$ becomes the ball of radius $x$ with the
same defining property and the disk $\Omega$ becomes the ball of radius $R(y)$. 
All arguments and relevant bounds still hold since they only rely on the triangle inequality;
the verification is left to the reader. Consequently, the approximation guarantee remains
the same. 

\smallskip
2. It is apparent from the context that our methods extend to a broader class of neighborhoods,
namely those that are approximable within a prescribed accuracy by unions of polyhedra
(this class includes curved objects, for instance balls of arbitrary radii).

\paragraph{An almost tight example.}
Let $\Delta{abc}$ be an isosceles triangle with $|ca|=|cb|=1-\eps$, $|ab|=1$,
for a small $\eps>0$, \eg, set $\eps=1/(n-1)$. Let $\N=\{X_1,\ldots,X_n\}$, where
$X_1=ac$, $X_2=bc$, and $X_3,\ldots,X_n$ are $n-2$ points at distance $1-\eps$ from $c$,
below $ab$ and whose projections onto $ab$ are close to the midpoint of $ab$; see
Figure~\ref{fig:eg}. Note that $ab$ is the unique (bichromatic) diameter of $V$
(the set of vertices in $\N$).  {\tt Algorithm A2} selects $a \in X_1$, $b \in X_2$,
and $X_3,\ldots,X_n$. 

\begin{figure}[hbtp]
  \centering
  \includegraphics[scale=0.5]{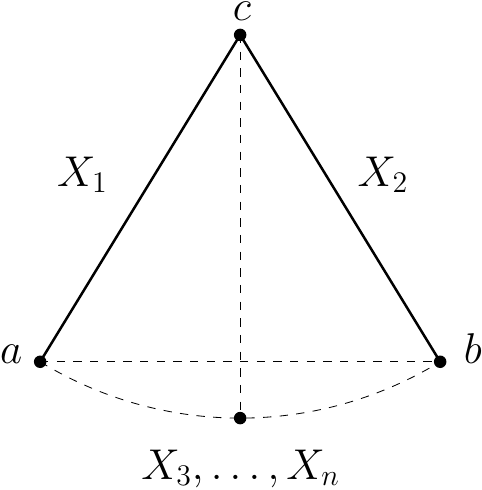}
\caption{A tight example.}
\label{fig:eg}
\end{figure}

The edge $ab$ has unit length and each of the remaining edges has length close to
$\sqrt{2-\sqrt3}$. Therefore,  the spanning tree constructed by {\tt Algorithm A2}
is of length close to
$1 + \sqrt{2-\sqrt3} (n-2)$, while the longest spanning tree has length
at least $(1-\eps)(n-1)=n-2$. As such, the approximation ratio of {\tt Algorithm A2}
approaches $\sqrt{2-\sqrt3}=0.517\ldots $ for large $n$.
Note that this is a tight example for the case $y \leq 0$, for which the ratio of {\tt Algorithm A2}
is at least $\sqrt{2-\sqrt3}$; and an almost tight example in general, since the overall
approximation ratio of {\tt Algorithm A2} is $0.511$.
Moreover, the example shows that every algorithm that always includes a bichromatic
diameter pair in the solution (as the vertices of the corresponding neighborhoods)
is bound to have an approximation ratio at most $\sqrt{2-\sqrt3}$.

\paragraph{Time complexity of {\tt Algorithm A2}.}
It is straightforward to implement the algorithm to run in quadratic time
for any fixed $d$.  All interpoint distances can be easily computed in $O(N^2)$ time.
Similarly the farthest point from $o$ in each neighborhood (over all neighborhoods)
can all be computed in $O(N)$ time. 
Subquadratic algorithms for computing the diameter and farthest bichromatic pairs
in higher dimensions can be found in~\cite{AMS92,BT83,Ra01,Ro93,TM82}; see also
the two survey articles~\cite{Ep00,Mi00}.

\section{Conclusion} \label {sec:conclusion}

We gave two simple approximation algorithms for {\sc Max-St-N}: one with ratio $0.5$
and one with ratio $0.511$
(the latter with a slightly more elaborate analysis but equally simple principles).
The first algorithm outputs a star centered at one of the endpoints of a bichromatic diameter.
The second algorithm outputs either a star centered at one of the endpoints of a monochromatic
diameter or a 2-star with the endpoints of a bichromatic diameter as its centers.
A 2-\emph{star} with \emph{centers} $a, b$ consists of an edge connecting $a, b$
and $n-2$ edges connecting every other vertex with one of the two centers.

The following variants represent extensions of the Euclidean maximum TSP for the
neighborhood setting. 
In the \emph{Euclidean Maximum Traveling Salesman Problem}, given a set of points
in the Euclidean space $\RR^d$, $d \geq 2$,
one seeks a cycle (a.k.a. \emph{tour}) that visits these points (as vertices) 
and has maximum length; see~\cite{BFJ+03}. 
In the \emph{Maximum Traveling Salesman Problem with Neighborhoods} ({\sc Max-Tsp-N}),
each point is replaced by a point-set, called  \emph{neighborhood} (or \emph{region}),
and the cycle must connect $n$ representative points, one chosen from each neighborhood
(duplicate representatives are allowed), and the cycle has maximum length. 
Since the original variant with points is NP-hard when $d \geq 3$ (as shown in~\cite{BFJ+03}),
the variant with neighborhoods is also NP-hard for $d \geq 3$.
The complexity of the original problem in the plane is unsettled, although the problem
is believed to be NP-hard~\cite{Fe99}.
In the \emph{path} variant, one seeks a path of maximum length. 

\smallskip
The following problems remain open for future investigation:
\begin{enumerate} \itemsep 2pt
\item What is the computational complexity of {\sc Max-St-N}?
\item Can a better approximation be obtained by constructing candidate spanning trees
  in the form of a $3$-star? (A $3$-\emph{star} with \emph{centers} $a,b,c$ consists of
  two edges connecting $a,b,c$ and $n-3$ edges connecting every other vertex with one
  of the $3$ centers.)
\item What approximations can be obtained for the \emph{cycle} or \emph{path} variants of
  {\sc Max-Tsp-N}?
\end{enumerate}


\begin{thebibliography}{99}

\bibitem {AMS92} P. K. Agarwal, J. Matou{\v s}ek and S. Suri,  
Farthest neighbors, maximum spanning trees and related problems in higher dimensions,
\emph{Computational Geometry: Theory and Applications} \textbf{1} (1992), 189--201.

\bibitem {ARS95} N. Alon, S. Rajagopalan and S. Suri,
Long non-crossing configurations in the plane,
{\em Fundamenta Informaticae} {\bf 22} (1995), 385--394.

\bibitem{AH94}
E.~M. Arkin and R.~Hassin,
Approximation algorithms for the geometric covering salesman problem,
{\em Discrete Applied Mathematics} \textbf{55} (1994), 197--218.

\bibitem{BFJ+03}
A.~I.~Barvinok, S.~Fekete, D.~S.~Johnson, A.~Tamir, G.~J.~Woeginger, and R.~Woodroofe,
The geometric maximum traveling salesman problem,
\emph{Journal of the ACM} \textbf{50(5)} (2003), 641--664.

\bibitem{BE97} M. Bern and D. Eppstein,
Approximation algorithms for geometric problems,
in \emph{Approximation Algorithms for NP-hard Problems}
(D. S. Hochbaum, editor),
PWS Publishing Company, Boston, MA, 1997, pp.~296--345.

\bibitem {BT83} B. K. Bhattacharya and G. T. Toussaint,
Efficient algorithms for computing the maximum distance between two finite planar sets, 
\emph{Journal of Algorithms} \textbf{4} (1983), 121--136.

\bibitem {BBC+18} 
A. Biniaz, P. Bose, K. Crosbie, J.-L. De Carufel, D. Eppstein, A. Maheshwari, and M. Smid,
Maximum plane trees in multipartite geometric graphs,
\emph{Algorithmica} \textbf{81(4)} (2019), 1512--1534.

\bibitem {DFH+15} R. Dorrigiv, R. Fraser, M. He, S. Kamali, A. Kawamura,
  A. L{\'{o}}pez{-}Ortiz, and D. Seco,
On minimum- and maximum-weight minimum spanning trees with neighborhoods,
\emph{Theory of Computing Systems} \textbf{56(1)} (2015), 220--250.

\bibitem{DT10}
A. Dumitrescu and Cs. D. T\'oth,
Long non-crossing configurations in the plane,
\emph{Discrete and Computational Geometry} \textbf{44(4)} (2010), 727--752.

\bibitem {Ep00} D. Eppstein,
Spanning trees and spanners,
in \emph{Handbook of Computational Geometry}
(J.-R.~Sack and J.~Urrutia, editors),
Elsevier Science, Amsterdam, 2000, pp.~425--461.

\bibitem {Fe99} S.~Fekete,
Simplicity and hardness of the maximum traveling salesman problem under geometric distances,
\emph{Proc. 10th Annual ACM-SIAM Symposium on Discrete Algorithms},
ACM/SIAM, 1999, pp.~337--345. 


\bibitem {Mi00} J.~S.~B. Mitchell,
Geometric shortest paths and network optimization,
in \emph{Handbook of Computational Geometry}
(J.-R. Sack and J.~Urrutia, editors),
Elsevier Science, Amsterdam, 2000, pp.~633--701.

\bibitem{Mi17} J.~S.~B. Mitchell,
Shortest paths and networks,
in \emph{Handbook of Computational Geometry}
  (J.~E.~Goodman, J.~O'Rourke, and C.~D.~T\'oth, editors),
3rd edition, CRC Press, Boca Raton, FL, 2017, pp.~811--848.

\bibitem {MPSY90} C. Monma, M. Paterson, S. Suri and F. Yao,
Computing {Euclidean} maximum spanning trees,
\emph{Algorithmica} \textbf{5} (1990), 407--419.

\bibitem{PA95}
J. Pach and P.~K.~Agarwal,
\emph{Combinatorial Geometry},
John Wiley, New York, 1995.

\bibitem{PS85} F. P. Preparata and M. I. Shamos,
\emph{Computational Geometry},
Springer-Verlag, New York, 1985.

\bibitem{Ra01} E. A. Ramos,
An optimal deterministic algorithm for computing the diameter of a three-dimensional point set,
\emph{Discrete and Computational Geometry} \textbf{26(2)} (2001), 233--244.

\bibitem {Ro93} J. M. Robert,
Maximum distance between two sets of points in $\mathbb{R}^d$,
\emph{Pattern Recognition Letters} \textbf{14} (1993), 733--735.

\bibitem {TM82} G. T. Toussaint and M. A. McAlear,
A simple $O(n \log n)$ algorithm for finding the maximum distance 
between two finite planar sets, 
\emph{Pattern Recognition Letters} \textbf{1} (1982), 21--24.

\bibitem {YLXX07} Y. Yang, M. Lin, J. Xu, and Y. Xie,
Minimum spanning trees with neighborhoods,
In \emph{Proc. 3rd Int. Conf. on Algor. Aspects in Information and Management} (AAIM),
vol.~4508 of LNCS, Springer, 2007, pp.~306--316.

\end{thebibliography}
\end{document}